\documentclass[journal, 10pt]{IEEEtran}
\usepackage{amsmath}                
\usepackage{amssymb}  
\usepackage{psfrag,graphicx}
\usepackage{epsfig}
\usepackage{tikz}
\usepackage{multicol}
\usepackage{blindtext}        
\newtheorem{theorem}{\textbf{Theorem}}

\newtheorem{props}{\textbf{Proposition}}
\newtheorem{defin}{\textbf{Definition}}

\newtheorem{problem}{\textbf{Problem}}

\def\R{\mathbb{R}}
\def\E{\mathbf{E}}

\def\F{\mathbf{F}}
\def\G{\mathbf{G}}
\def\H{\mathbf{H}}
\def\B{\mathbf{B}}

\def\N{\mathcal{N}}
\def\QED{~\rule[-1pt]{6pt}{6pt}\par\medskip}
\newenvironment{proof}{{\bf Proof.\ }}{ \hfill \QED}  

\title{\LARGE \bf Kalman meets Shannon: Optimal Linear Estimation over Gaussian Communication Channels}
\author{Ather Gattami\\
Ericsson Research\\
F\"ar\"ogatan 6, Stockholm, Sweden\\ 
E-mail: ather.gattami@ericsson.com\\[8mm]
}

\begin{document}
\maketitle

\begin{abstract}
We consider the problem of communicating the state of a dynamical system via
a Shannon Gaussian channel. The receiver,
which acts as both a decoder and estimator, observes the noisy measurement of the channel
output and makes an optimal estimate of the state of the dynamical system in
the minimum mean square sense. The transmitter observes a possibly noisy measurement of the state of the dynamical system. These measurements are then used to encode the message to be transmitted over a noisy Gaussian channel, where a \textit{per sample} power constraint is imposed on the transmitted message.
Thus, we get a mixed problem of Shannon's source-channel coding problem and a sort of Kalman filtering problem.
We first consider the problem of communication with full state measurements at the transmitter and show that optimal linear encoders don't need to have memory and the optimal linear decoders have an order of at most that of the state dimension. We also give explicitly the structure of the optimal linear filters. For the case where the transmitter has access to noisy measurements of the state, we derive a separation principle for the optimal communication scheme, where the transmitter needs a filter with an order of at most the dimension of the state of the dynamical system.  The results are derived for first order linear dynamical systems, but may be extended to MIMO systems with arbitrary order.
\end{abstract}


\section{Introduction}
\subsection{Background}
This paper studies the problem of communicating the state of a dynamical system via
a Shannon Gaussian channel. The receiver,
which acts as both a decoder and estimator, observes the noisy measurement of the channel
output and makes an optimal estimate of the state of the dynamical system in
the minimum mean square sense. The transmitter observes a possibly noisy measurement of the state of the dynamical system. These measurements are then used to encode the message to be transmitted over a noisy Gaussian channel, where a \textit{per sample} power constraint is imposed on the transmitted message.

Shannon (\cite{shannon:48, shannon1949}) considered the problem of reliable communication of a one-dimensional source over a one-dimensional Gaussian channel. In particular, Shannon considered the  following coding-decoding setting for an analog Gaussian channel:
$$
\inf_{\substack{ f:\R\rightarrow \R\\ g:\R\rightarrow \R\\ \E |g(x)|^2 \leq P}}
\E|x - f(g(x)+n)|^2 
$$
where $x\sim \mathcal{N}(0,X)$ and $n\sim \mathcal{N}(0,N)$. Shannon showed
that the infimum can be attained by using linear encoder and decoder $g$ and $f$, respectively.


\begin{figure}
	\center
  	\includegraphics[width = 1\columnwidth]{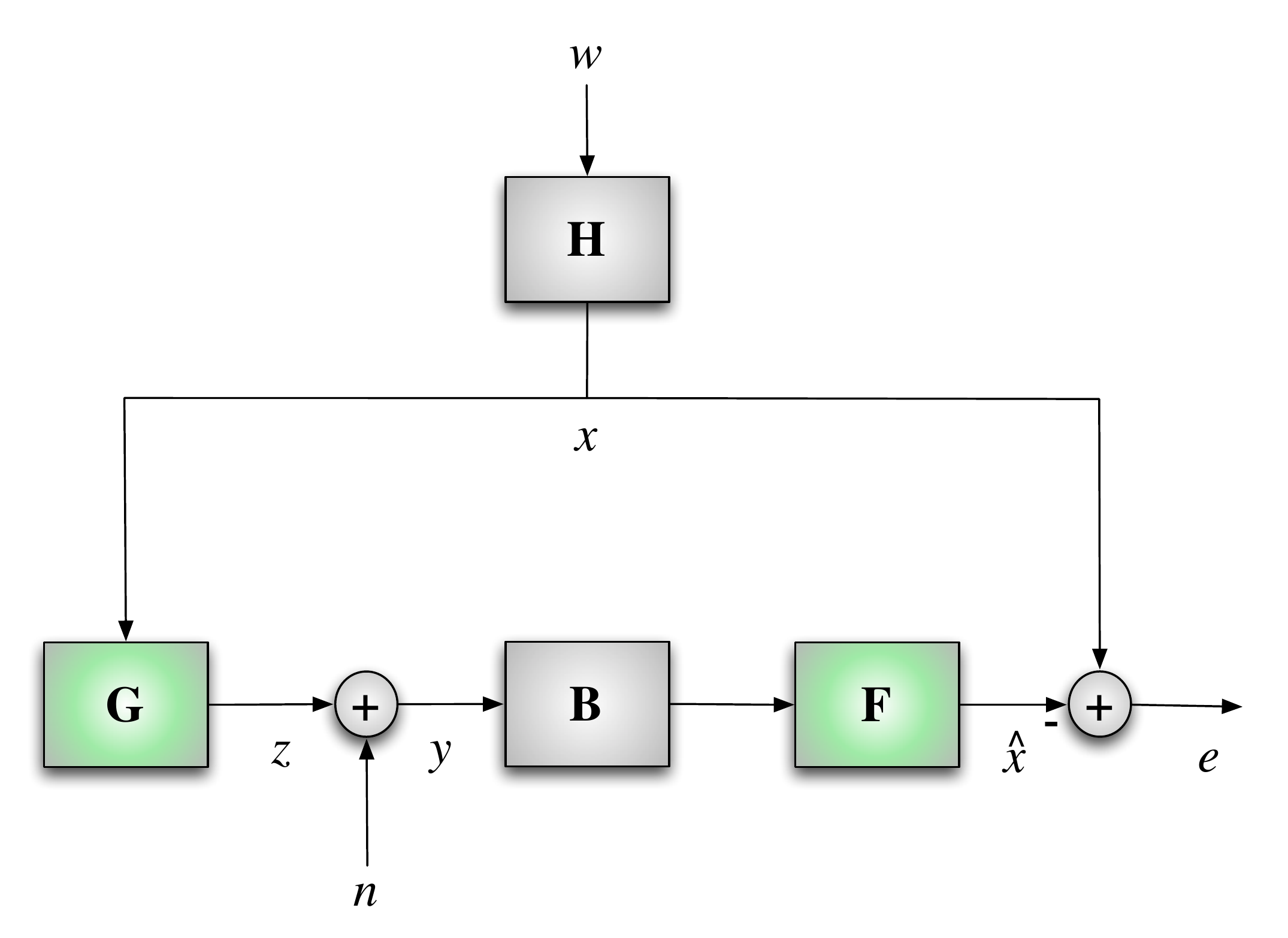}	
	\caption{A simple model of an estimation problem of the state of the dynamical system $\H$ over a Gaussian communications channel with Gaussian noise $n\sim\mathcal{N}(0,N)$, and delay given by the backward shift operator 
	$\B$. The optimization parameters are given by the encoder $\G$ and the decoder $\F$. The samples of the encoder output $z$ are power limited with $\E|z(t)|^2 \leq P$. }
	\label{mainFig}
\end{figure}


More specifically, consider the block-diagram in Fig. \ref{mainFig}. We have the process noise given by $w$, which is assumed to be Gaussian white noise, and the state
is given by $x = \H w$ where $\H$ is a causal linear operator/filter. \\
The precoder is given by the causal operator
$\G$, not necessarily linear. The encoded signal $z = \G x$ is then transmitted over a Gaussian channel with white noise given by $n$. Typically, one has power constraints on the transmitted signal $z(t)$, that is $\E |z(t)|^2 \leq P$, for some positive real number $P$. At the other end, the message received is $y(t) = z(t)+n(t)$, for $t=0, ..., T-1$, and is delayed with one time step by the backward shift operator $\B$. Finally, the causal operator $\F$ is the decoder, designed to reconstruct the state $x$ by $\hat{x}=\F \B y$, to minimize the mean squared error $\E|e|^2 = \E |x-\hat{x}|^2$. 

For the case where $\G$ is a fixed linear operator, the optimal filter $\F$ is well known to be given by the optimal Kalman filter, which is a linear operator. However, if $\G$ is a precoder to be co-designed together with $\F$, we get a nonconvex problem even if we restrict the optimization problem to be carried out over linear operators/filters. Prior to this work, it has not been known whether the order of the linear optimal filters is finite and what the upper bound is.

\subsection{Previous work}
\cite{kalman:1960} made a fundamental contribution to optimal control and filtering of linear dynamical systems by deriving recursive state space solutions. The model considered by Kalman assumes linear measurements of the state, possibly partial and corrupted by noise. The role of a communication channel \textit{with feedback} and its affect on stability was studied in \cite{tatikonda:2004} and necessary conditions for stability were given. Fundamental limitations of performance were studied in \cite{martins:2008}. In \cite{charalambous:2008}, the capacity was matched to achieve a certain distortion, which required a time-varying power constraint that cannot be fixed beforehand. 
The problem of \textit{linear} communication and filtering over a noisy channel for the stationary case has been considered in \cite{joh10acc} where it was shown that this problem can be transformed to a convex optimization problem that grows with the size of the time horizon. However, the order of the  linear optimal filters obtained from \cite{joh10acc} is infinite. This paper presents similar results to \cite{gattami:ifac:2014} with complete proofs.

\subsection{Contributions}
We consider the linear dynamical system $\mathbf{H}$ given by
\begin{align*} 
	x(t+1) 	&= ax(t) + bw(t) \\
	\gamma(t)	&= cx(k) + dv(t)
\end{align*}

The main contribution of this paper is to show that the optimal linear filters $\F$ and $\G$ in the communication scheme described in Fig. \ref{mainFig}(the figure shows the case $c=1$, $d=0$) have a finite order independent of the size of the time horizon. 
We also show explicitly the structure of the optimal filters, which is given by

\begin{equation*}
 \G: \hspace{0.5cm} 
 \left\{\begin{aligned}
 	\breve{x}(t) &= \E\{x(t) | \gamma^t\}\\		
	z(t) &= \frac{\sqrt{P}}{\sigma_t} \breve{x}(t), ~~\sigma_t^2=  \E|\breve{x}(t)|^2\\
 \end{aligned}\right. 
\end{equation*}
\begin{equation*}
 \hspace{-1.4cm} \F: \hspace{0.8cm} 
 \begin{aligned}
	\hat{x}(t) 	&= \E\{x(t) | y^{t-1}\}
 \end{aligned}
 \end{equation*}
The interpretation is that the transmitter estimates the state of the dynamical system which is given by a Kalman filter, and then transmit a scaled version of that estimate to satisfy the given power constraint.
The receiver's optimal strategy is to estimate the \textit{transmitter's estimate} of the state of the dynamical system. Thus, the order of the filters is bounded by the dimension of the state and not the time horizon.

\subsection{Notation}
\begin{tabular}{ll}
$x^t$		& $x^t = (x(0), x(1), ..., x(t))$.\\
$\mathbb{L}$ & The set of lower triangular matrices.\\
$\B$  	& Denotes the backward shift operator,\\ 
		& $x(t-1) = \B x(t)$.\\			
$\mathbf{E}\{\cdot\}$ 	& $\mathbf{E} \{x\}$ denotes the expected value of the\\ 
			& stochastic variable $x$.\\
$\mathbf{E}\{\cdot|\cdot\}$ 	& $\mathbf{E} \{x|y\}$ denotes the expected value of the\\
			& stochastic variable $x$ given $y$.\\
$\mathbf{cov}$ & $\mathbf{cov}\{x,y\} = \E\{xy^\intercal \}$.\\
$h(x)$ 		& Denotes the entropy of $x$.\\
$h(x|y)$ 		& Denotes the entropy of $x$ given $y$.\\
$I(x;y)$ 		& Denotes the mutual information between\\
			& $x$ and $y$.\\

$\mathcal{N}(m,V)$  & Denotes the set of Gaussian variables with\\ 
				& mean $m$ and covariance $V$.\\
\end{tabular}


\section{Preliminaries}

\begin{defin} 
The entropy of a real-valued stochastic variable $x$ with probability distribution $p(\cdot)$ is defined as
$$
h(x) = -\int_{-\infty }^{\infty } p(z) \log_2 p(z) dz
$$
\end{defin}

\begin{defin}
For two real valued stochastic variables $x$ and $y$, the conditional entropy of $x$ given $y$ is
defined as
$$
h(x|y) = h(x,y) - h(y).
$$
\end{defin}

\begin{defin}
The mutual information between $x$ and $y$ is defined as
$$
 I(x; y) = h(x) - h(x|y) = h(y) - h(y|x).
$$
\end{defin}

\begin{props}[Entropy Power Inequality]
\label{entpineq}
	If $x$ and $y$ are independent scalar random variables, then
	$$
	2^{2h(x+y)} \geq 2^{2h(x)} + 2^{2h(y)}
	$$
	with equality if $X$ and $Y$ are Gaussian stochastic variables. 
\end{props} 
\begin{proof}
	See \cite{cover:2006}, p. 674 - 675.
\end{proof}

\begin{defin}
	Random variables $x, y, z$ are said to form a Markov chain in that order if the conditional 
	distribution of $z$ depends only on $y$ and conditionally independent of $x$. This is denoted by
	$x\rightarrow y \rightarrow z$.
\end{defin}

\begin{props}[Data-Processing Inequality]
\label{dataproc}
	If $$x\rightarrow y \rightarrow z,$$ 
	then $$I(x; z)\leq I(y; z).$$	
\end{props}
\begin{proof}
	See \cite{cover:2006}, p. 34-35.
\end{proof}

\begin{props}
\label{lse}
Let  $x$ and $y$ be two stochastic variables. The optimal solution to the optimization problem
$$
\inf_{f(\cdot)} \E |x-f(y)|^2
$$
is unique and given by the expectation of $x$ given $y$
\[
\begin{aligned}
f_\star(y) = \E\{x | y\} .
\end{aligned}
\]
Furthermore, $f_\star(y)$ and $x-f_\star(y)$ are uncorrelated.  

\end{props}
\begin{proof}
Consult (\cite{shiryaev}, p. 237).
\end{proof}
\begin{props}
\label{logerror}
Consider the stochastic variables $x$ and $y$, and let the estimation error of $x$ based on $y$
be given by
$$
\tilde{x} = x - \E\{x|y\}.
$$
Then,
\begin{equation}
	\label{lowerb}
		\frac{1}{2}\log_2{(2\pi e \E \{\tilde{x}^2 \})} \geq h(x|y) = h(\tilde{x})
\end{equation}
with equality if and only if $x$ and $y$ are jointly Gaussian. 
\end{props}  
\begin{proof}
Consult \cite{gamal:nit}, p. 21.
\end{proof}


\section{Problem Formulation}
\subsection{Optimal Filtering over Noisy Communication Channel with Full State Measurements at The Transmitter}

Let $\mathbf{H}$ be a first order linear time invariant dynamical system  with state-space realization

\begin{align} 
\label{linearsystem}
	x(t+1) 	&= ax(t) + bw(t) , ~~~~~~~ x(0)=0, ~~~0\leq t\leq T-1,
\end{align}
where $a,b \in\R$ 
and $w$ is assumed to be white Gaussian noise with $w(t)\sim \N(0,1)$ for all $0\leq t\leq T-1$. 

The precoder is a map $\G: x^t\mapsto z(t)$, where $z$ is the signal transmitted over the Gaussian channel. We have a power constraint on the transmitted signal $z(t)$ given by 
$\E |z(t)|^2 \leq P$. 

The measurements at the decoder are given by $y(0) :=0$ and 
$$
y(t) = z(t)+n(t), ~~~~~ \text{for }t\geq 1,
$$
where $n$ is a Gaussian white noise process with  $n(t)\sim \mathcal{N}(0, N)$. The decoder is a map
$\F: y^{t-1}\mapsto \hat{x}(t)$.

The objective is to design causal precoder and decoder maps $\G: x^t\mapsto z(t)$ and $\F: y^{t-1}\mapsto \hat{x}(t)$, respectively,  such that the average of the mean squared error is minimized:
$$
\frac{1}{T}\sum_{t=1}^T \E|x(t)-\hat{x}(t)|^2 \rightarrow \min.
$$
The precoder and decoder maps can be equivalently written as a causal dynamical system according to
 \begin{equation}
 \label{nonlin0}
 \begin{aligned}
 z(t) 		&= g_t(x^t), ~~~ \E\{z^2(t)\}\leq P\\
 \hat{x}(t) 	&= f_t(y^{t-1}),
 \end{aligned}
 \end{equation}
where $g_t$ is the precoder and $f_t$ is the decoder.

Now we may formalize our first problem statement:

\begin{problem}
\label{p1}
Consider the linear dynamical system
\begin{align*} 
	x(t+1) 	&= ax(t) + bw(t) , ~~~~~~~ x(0)=0, ~~~0\leq t\leq T-1,
\end{align*}
where $a,b \in\R$ 
and $w(t)\sim \N(0,1)$ for $0\leq t\leq T-1$.
Let $n$ be a Gaussian white noise process independent of $w$, with  $n(t)\sim \mathcal{N}(0, N)$. Find an optimal precoder and decoder pair (\ref{nonlin0}) such that
$$
\frac{1}{T}\sum_{t=1}^T \E|x(t)-\hat{x}(t)|^2 \rightarrow \min,
$$
where $y(0) = 0$ and
$
y(t) = z(t)+n(t),  ~
\text{for } t\geq 1.
$
 \end{problem}


\subsection{Linear optimal precoder/decoder design}
The linear filter $\H$ has the following Toeplitz matrix representation over the time 
$t = 1, ..., T$:
\begin{align} 
	\begin{bmatrix}
		x(1)\\ 	
		x(2)\\
		x(3)\\
		\vdots \\
		x(T)
	\end{bmatrix}
	&= 
	 \begin{bmatrix}
  		b 		& 0 		& 0 	 	&\cdots 	& 0 \\
  		ab 		& b 		& 0	 	&\cdots	& 0\\
  		a^2b 	& ab		& b 	 	&\cdots 	& 0 \\
  		\vdots 	& \vdots	& \vdots 	&\ddots 	& \vdots \\
		a^Tb		& a^{T-1}b& a^{T-2}b&\cdots 	& b 
 	\end{bmatrix} 
	\begin{bmatrix}
		w(0)\\ 	
		w(1)\\
		w(2)\\
		\vdots \\
		w(T-1)
	\end{bmatrix}
\end{align}

Let the precoder $\G$ be a causal linear filter that maps
$x$ to $z$:
\begin{align} 
	\begin{bmatrix}
		z(1)\\ 	
		z(2)\\
		z(3)\\
		\vdots \\
		z(T)
	\end{bmatrix}
	&= 
	 \begin{bmatrix}
  		G_{11} 		& 0 		& 0 	 	&\cdots 	& 0 \\
  		G_{21} 		& G_{22} 		& 0	 	&\cdots	& 0\\
  		G_{31} 	& G_{32}		& G_{33} 	 	&\cdots 	& 0 \\
  		\vdots 	& \vdots	& \vdots 	&\ddots 	& \vdots \\
		G_{T1}	& G_{T2}& G_{T3}&\cdots 	& G_{TT} 
 	\end{bmatrix}
	\begin{bmatrix}
		x(1)\\ 	
		x(2)\\
		x(3)\\
		\vdots \\
		x(T)
	\end{bmatrix}
\end{align}

The precoder is subject to a power constraint on its output signal
$z=\G x$ given by $\E|z(t)|^2\leq P$, for $t= 0, ..., T.$ 

The decoder  $\F$ is a causal linear filter that observes the delayed measurements with $y(0):=0$ and $y(t) = z(t)+n(t)$ for $t\geq 1$. It has the following linear operator representation:
\begin{align} 
	\begin{bmatrix}
		\hat{x}(1)\\ 	
		\hat{x}(2)\\
		\hat{x}(3)\\
		\vdots \\
		\hat{x}(T)
	\end{bmatrix}
	&= 
	 \begin{bmatrix}
  		F_{11} 		& 0 		& 0 	 	&\cdots 	& 0 \\
  		F_{21} 		& F_{22} 		& 0	 	&\cdots	& 0\\
  		F_{31} 	& F_{32}		& F_{33} 	 	&\cdots 	& 0 \\
  		\vdots 	& \vdots	& \vdots 	&\ddots 	& \vdots \\
		F_{T1}	& F_{T2}& F_{T3}&\cdots 	& F_{TT} 
 	\end{bmatrix}
	\begin{bmatrix}
		y(0)\\ 	
		y(1)\\
		y(2)\\
		\vdots \\
		y(T-1)
	\end{bmatrix}
\end{align}
The output $\hat{x} = \F \B y$ is the optimal estimate of $x$ in the sense that the average of the mean squared error, is minimized:
$$
\frac{1}{T}\sum_{t=1}^T \E|x(t)-\hat{x}(t)|^2.
$$
Now let 
\begin{align} 
H
 =
 \begin{bmatrix}
  		b 		& 0 		& 0 	 	&\cdots 	& 0 \\
  		ab 		& b 		& 0	 	&\cdots	& 0\\
  		a^2b 	& ab		& b 	 	&\cdots 	& 0 \\
  		\vdots 	& \vdots	& \vdots 	&\ddots 	& \vdots \\
		a^T b		& a^{T-1}b& a^{T-2}b  &\cdots 	& b 
 	\end{bmatrix} ,
\end{align}
\begin{align} 
	G
	 =
	 \begin{bmatrix}
  		G_{11} 		& 0 		& 0 	 	&\cdots 	& 0 \\
  		G_{21} 		& G_{22} 		& 0	 	&\cdots	& 0\\
  		G_{31} 	& G_{32}		& G_{33} 	 	&\cdots 	& 0 \\
  		\vdots 	& \vdots	& \vdots 	&\ddots 	& \vdots \\
		G_{T1}	& G_{T2}& G_{T3}&\cdots 	& G_{TT} 
 	\end{bmatrix} ,
\end{align}
\begin{align} 
	F
	 =
	 \begin{bmatrix}
  		F_{11} 		& 0 		& 0 	 	&\cdots 	& 0 \\
  		F_{21} 		& F_{22} 		& 0	 	&\cdots	& 0\\
  		F_{31} 	& F_{32}		& F_{33} 	 	&\cdots 	& 0 \\
  		\vdots 	& \vdots	& \vdots 	&\ddots 	& \vdots \\
		F_{T1}	& F_{T2}& F_{T3}&\cdots 	& F_{TT} 
 	\end{bmatrix} ,
\end{align}

$$
x
=
	\begin{bmatrix}
		x(1)\\ 	
		x(2)\\
		x(3)\\
		\vdots \\
		x(T)
	\end{bmatrix}, \hspace{3mm}
w
 = 
	\begin{bmatrix}
		w(0)\\ 	
		w(1)\\
		w(2)\\
		\vdots \\
		w(T-1)
	\end{bmatrix},  \hspace{3mm}
z
 = 
	\begin{bmatrix}
		z(1)\\ 	
		z(2)\\
		z(3)\\
		\vdots \\
		z(T)
	\end{bmatrix},
$$
$$
~~~~~
\hat{x} = 
	\begin{bmatrix}
		\hat{x}(1)\\ 	
		\hat{x}(2)\\
		\hat{x}(3)\\
		\vdots \\
		\hat{x}(T)
	\end{bmatrix}, \hspace{3mm} 	
n = 
	\begin{bmatrix}
		n(0)\\ 	
		n(1)\\
		n(2)\\
		\vdots \\
		n(T-1)
	\end{bmatrix},  \hspace{3mm} 
	y= 
	\begin{bmatrix}
		y(0)\\ 	
		y(1)\\
		y(2)\\
		\vdots \\
		y(T-1)
	\end{bmatrix}.	
$$

Then,
$$
x = H w, ~~~~~  z = GHw, ~~~~~ y = GHw+n, ~~~~~ \hat{x} = Fy,
$$
$$
\sum_{t=1}^T \E|x(t)-\hat{x}(t)|^2 = \E |x-\hat{x}|^2 =  \E|H w - F y|^2
$$

%

After some algebra, the least mean square error for a linear precoder and decoder will be given by
\begin{equation}	
\label{linearcoding}
\inf_{\substack{ G, F\in\mathbb{L}\\ G_t H
 H^* G_t^* \leq P}} \E|H w - F(G H w + n)|^2 
\end{equation}
Note that the optimization problem above is inherently non-convex, since 
we have a coupling term between $G$ and $F$ in the quadratic objective function.


 \section{Main Results}

\subsection{Optimal Transmission Scheme with Full State Information at The Transmitter}
The first result of this paper presents the structure of the optimal precoder and
decoder:

 \begin{theorem}
\label{main} 
The optimal linear communication scheme in Problem \ref{p1} is given by
  \begin{equation}
 \label{nofbeq}
 \begin{aligned}
	 \hat{x}(t) 	&= \E\{x(t) | y^{t-1}\}\\
 		z(t) 		&= \frac{\sqrt{P}}{\sigma_t} x(t),\\
 \end{aligned}
 \end{equation}
 where $\sigma_t^2=  \E|x(t)|^2$, for $t= 1, ..., T$.
\end{theorem}
 \begin{proof}
The proof is differed to the appendix.
\end{proof}
 
 The theorem above implies that the optimal filters are finite. Clearly, the filter $\G$ is static, and thus, has no memory, and $\F$ is simply the Kalman filter that estimates the state of the linear process (\ref{linearsystem}) given the output measurement $$y(t)= \frac{\sqrt{P}}{\sigma_t} x(t)+n(t).$$
 

\subsection{Time-Varying Systems}

The results considered so far treated the case where the state stems from a linear time invariant system. It's straight forward to verify that the results hold when we replace the parameters static $a, b, P, N$ with time varying parameters $a(t), b(t), P(t), N(t)$, respectively.


\subsection{Separation Principle for Optimal Communication}
Consider the linear system
\begin{align*} 
	x(t+1) 		&= ax(t) + bw(t) \\
	\gamma(t)	&= cx(k) + dv(t)
\end{align*}
for $0\leq t\leq T-1$, with $x(0)=0$ and $(w(t),v(t))$ is white Gaussian noise process with a given covariance. We assume now that the transmitter does't have
access to the state $x(t)$ but $\gamma(t)$ instead. We get the following problem.

\begin{problem}
\label{p2}
Consider the linear system
\begin{align*} 
	x(t+1) 	&= ax(t) + bw(t) \\
	\gamma(t)	&= cx(k) + dv(t)
\end{align*}
$x(0)=x_0$, $0\leq t\leq T-$1, where $a,b \in\R$, and
$$
\E 
\begin{bmatrix}
	w(t)\\
	v(t)	
\end{bmatrix}
\begin{bmatrix}
	w(t)\\
	v(t)	
\end{bmatrix}^\intercal =
\begin{bmatrix}
	V_{ww}(t) & V_{wv}(t)\\
	V_{vw}(t) & V_{vv}(t)	
\end{bmatrix}
$$
is given for $0\leq t\leq T$.
Let $n$ be a white Gaussian noise process independent  of $w$, with $n(t)\sim \mathcal{N}(0, N)$. Find an optimal precoder and decoder pair 
\begin{equation}
 \label{nonlin2}
 \begin{aligned}
 z(t) 		&= g_t(\gamma^t)\\
 y(t)		&= z(t)+n(t)\\
 \hat{x}(t) 	&= f_t(y^{t-1})
 \end{aligned}
 \end{equation}
 such that
$$
\frac{1}{T}\sum_{t=1}^T \E|x(t)-\hat{x}(t)|^2 \rightarrow \min,
$$
where $y(0) = 0$.
 \end{problem}

The optimal linear transmission scheme is for the transmitter to find the best estimate of $x(t)$ based on $\gamma^t$, namely $\breve{x}(t) = \E\{x(t) | \gamma^t\}$, and then use this estimate as the state to be transmitted using the optimal communication scheme for the case of full state measurement at the transmitter given by Theorem \ref{main}. 

 \begin{theorem}
 \label{ofb}
 \label{mainnofb} 
The optimal linear communication scheme in Problem \ref{p2} is given by
  \begin{equation}
 \label{nofbeq}
 \begin{aligned}
 	\breve{x}(t) &= \E\{x(t) | \gamma^t\}\\
	\hat{x}(t) 	&= \E\{\breve{x}(t) | y^{t-1}\}\\
 		z(t) 		&= \frac{\sqrt{P}}{\sigma_t} \breve{x}(t),\\
 \end{aligned}
 \end{equation}
 where $\sigma_t^2=  \E|\breve{x}(t)|^2$, for $t= 1, ..., T$,
 $$
\breve{x}(t+1) = a\breve{x}(t) + \beta(t) \omega(t)
 $$

with $\omega(t) \sim \mathcal{N}(0,1)$, $V_{\xi \xi}(0) = 0$,

\begin{equation*} 
\begin{aligned}
L(t)					&= V_{\xi \xi}(t)c(c^2V_{\xi \xi}(t) + d^2V_{vv}(t) )^{-1} \\ 
V_{\xi \xi}(t+1) 	&=  (a-aL(t)c)^2 V_{\xi \xi}(t) \\
					&+
\begin{bmatrix}
	b & -aL(t)
\end{bmatrix}
\begin{bmatrix}
	V_{ww}(t) & V_{wv}(t)\\
	V_{vw}(t) & V_{vv}(t)	
\end{bmatrix}
\begin{bmatrix}
	b & -aL(t)
\end{bmatrix}^\intercal
\end{aligned}
\end{equation*}

\begin{equation*} 
	\begin{aligned}
	\beta^2(t) &= L^2(t+1)  (c^2 V_{\xi \xi}(t+1)  + d^2 V_{vv}(t+1) )
	\end{aligned}
\end{equation*}
\end{theorem}
\begin{proof}
The proof is differed to the appendix.
\end{proof}

\subsection{Optimality of Linear Filters}
We have considered encoder-decoder design based on linear filters in the previous section, and it's interesting to examine whether linear filters are optimal among the class of all filters, linear and nonlinear. In the special case where we restrict the the filter $\G$ to have an output sample $z(t)$ such that $\hat{z}(t) =  \E\{z(t)  | y^{t-1} \}$ and $\tilde{z}(t) = z(t) - \hat{z}(t)$ are not only uncorrelated according to Proposition \ref{lse}, but also \textit{independet} (see the proof of Theorem \ref{main}). What this constraint means in practice and how restrictive it could be is not known, but it forces the optimal filters to be linear.

\section{Conclusions}
We considered the problem of optimal encoder-decoder filter design over a Shannon Gaussian channel under a \textit{per sample} power constraint, to transmit and estimate the state of a linear dynamical system.  
We first considered the problem of communication with full state measurements at the transmitter and showed that optimal linear encoders don't need to have memory and the optimal linear decoders have an order of at most that of the state dimension. We also showed explicitly the structure the optimal linear filters. For the case where the transmitter has access to noisy measurements of the state, we derived a separation principle for the optimal communication scheme, where the transmitter needs a filter with an order of at most the dimension of the state of the dynamical system.  The results were derived for first order linear dynamical systems, but may be generalized to MIMO systems with arbitrary order.

Future work includes the problem where noisy measurements are available from the decoder to the encoder through a noisy Gaussian feedback channel which is necessary in order to be able to track signals belonging to unstable systems as the time horizon goes to infinity.

\section{Acknowledgements}
The author is grateful to Prof. S. Yuksel and Dr. B. Lincoln for constructive and useful feedback.

\bibliography{../../ref/mybib}


\section*{Appendix}
 

\subsection*{Proof of Theorem \ref{main} }
	Suppose that $\E \{g_t(x^t)\} = \alpha_t$ where $\{\alpha_k\}_{k=0}^t$ are
	deterministic real numbers independent of $x^t$ and are known at the encoder $g_t$ 
	and decoder $f_t$. 
 Note that $y(t) = g_t(x^t)+n(t)$. The estimate of $x(t+1)$ based on $y(k)$, 
 $k=0, ..., t$, is the same as the estimate of $x(t+1)$ based on $y(k)-\alpha_k$ for $k=0, ..., t$ since $\alpha_k$ is deterministic and known at the decoder.
But it means that we can replace $g_t(x^t)$ with $g_t'(x^t)=g(x^t)-\alpha_t$, and $g_t'(x^t)$ satisfies both $\E \{g_t'(x^t)\} = 0$ and the power constraint $\E|g_t'(x^t)|^2 \leq P$ since
\begin{equation*}
 \begin{aligned}
\E|g_t'(x^t)|^2 	&= \E|g_t(x^t)-\alpha|^2 \\
			&= \E |g_t(x^t)|^2 - \alpha^2\\ 
			&= P-\alpha^2\leq P.
 \end{aligned}
 \end{equation*} 
 Thus, without loss of generality,  we may restrict the encoders $g$ to the set 

$$\{ g~ | ~\E \{g(x^t)\} = 0\}.$$ 

Let $\hat{x}(t|t) = f'_t(y^t)$ be the optimal estimate of $x(t)$ based on $y^t$ and  
let $\tilde{x}(t|t)  = x(t) - \hat{x}(t|t)$, for $t=0, ..., T$. We have that $f'_t(y^t) = \E\{x(t) | y^t\}$ according to Proposition \ref{lse}. Now we have that
\begin{equation}
 \label{nofbestimate1}
 	\begin{aligned}
		 \hat{x}(t|t) 	&= \E\{x(t)|y^t\}\\
		 				&=  \E\{(\hat{x}(t) + \tilde{x}(t)  | y^{t}\}\\
						&= \hat{x}(t) + \E\{\tilde{x}(t)  | y^t \},
 	\end{aligned}
 \end{equation} 
\begin{equation}
 \label{nofbestimate2}
 	\begin{aligned}
\tilde{x}(t+1) 	&= x(t+1) -\hat{x}(t+1) \\
				&= ax(t)+bw(t)-a\hat{x}(t|t)\\ &= a\tilde{x}(t|t) + bw(t)\\
\end{aligned}
 \end{equation} 
We see that minimizing $\E |\tilde{x}(t+1)|^2$ is equivalent to minimizing the mean square error of
$$
\tilde{x}(t|t) = \tilde{x}(t) - \E\{\tilde{x}(t)  | y^t \}
$$
at the decoder. Note that $\tilde{x}(t)$ and $y^{t-1}$ are jointly Gaussian. 
Since they are uncorrelated, they are also independent.

Recall that $z(t) = g_t(x^t)$. Let $\hat{z}(t) =  \E\{z(t)  | y^{t-1} \}$ and $\tilde{z}(t) = z(t) - \hat{z}(t)$. 

Now we have that
$$
P\geq \E{z^2(t)} = \E{ \hat{z}^2(t)} + \E{\tilde{z}^2(t)} 
$$
with equality if 
$$ \E{\tilde{z}^2(t)} = P-\E \hat{z}^2(t)=:  P'$$ 
Let $k$ be such that $k^2\E{\tilde{x}^2(t)} = P'$. 
The data processing inequality (Proposition \ref{dataproc}) together with the Shannon capacity of a Gaussian channel ( \cite{gallager}) gives an upper bound on the mutual information between $\tilde{x}(t)$ and $\tilde{z}(t) + n(t)$
 \begin{equation}
 \label{eq1}
	\begin{aligned}
 		I(\tilde{x}(t); \tilde{z}(t) + n(t))
											&\leq I(\tilde{z}(t); \tilde{z}(t) + n(t))\\
											&=	 \frac{1}{2} \log_2{\left(1+\frac{P'}{N}\right)}						
	\end{aligned}	
 \end{equation} 
 Equality holds if and only if $\tilde{z}(t) = k \tilde{x}(t)$, that is if and only if $z(t) = K(t)x(t)$ with
 $K(t) = \frac{\sqrt{P}}{\sigma_t}$, $\sigma_t^2 = \E|x(t)|^2$. 
 
 Equation (\ref{eq1}) is equivalent to
\begin{equation}
\label{nofbestimate6}
	2^{-2I(\tilde{z}(t); y(t))} = \frac{N}{P'+N}
\end{equation}
Now we get
\begin{eqnarray}
		2\pi e \E\{|\tilde{x}(t|t)|^2\}	
							&=& 2^{2h(\tilde{x}(t) | y^t)} \label{eq2} \\
							&=& 2^{2h(\tilde{x}(t) | z(t) + n(t), y^{t-1})} \label{eq3}\\
							&=& 2^{2h(\tilde{x}(t) | \hat{z}(t) + \tilde{z}(t) + n(t), y^{t-1})} \label{eq4}\\
							&=& 2^{2h(\tilde{x}(t) | \tilde{z}(t) + n(t))} \label{eq5}\\
							&=& 2^{2h(\tilde{x}(t)) - 2I(\tilde{x}(t);  \tilde{z}(t) + n(t))}  \label{eq6}\\
							&\geq& 2^{2h(\tilde{x}(t)) - 2I(\tilde{z}(t);  \tilde{z}(t) + n(t))} \label{eq7}\\
							&=& \frac{N}{P'+N} 2^{2h(\tilde{x}(t) )} \label{eq8}
 \end{eqnarray} 
where (\ref{eq2}) follows from Proposition \ref{logerror}(equality holds since $\tilde{x}(t)$ and $y^t$ are jointly Gaussian), (\ref{eq5}) follows from the fact that $y^{t-1}$ is independent of $(\tilde{x}(t), \tilde{z}(t), n(t))$,    
(\ref{eq6}) follows from the definition of mutual information, (\ref{eq7}) follows from (\ref{eq1}), and (\ref{eq8}) follows from (\ref{nofbestimate6}). Thus, the error 
$\E\{|\tilde{x}(t|t)|^2\}$ has a lower bound given by 
$$
\E\{|\tilde{x}(t|t)|^2\} \geq \frac{N}{2\pi e(P'+N)} 2^{2h(\tilde{x}(t) )}
$$
with equality if
$$z(t) = \frac{\sqrt{P}}{\sigma_t} x(t)$$
with $\sigma_t^2=  \E |x(t)|^2$. This completes the proof.

\subsection*{Proof of Theorem \ref{ofb}}

Define the estimates $\breve{x}(t) = \E\{x(t) | \gamma^t\}$ and 
$\breve{x}(t|t-1) = \E\{x(t) | \gamma^{t-1}\}$. Let
$$\xi(t) = x(t) - \breve{x}(t|t-1).$$ 
It's well known that $\breve{x}(t)$ is given by the Kalman filter
\begin{equation} 
\label{transs}
\begin{aligned}
	\breve{x}(t)		&= \breve{x}(t|t-1) + L(t)(c\xi(t)+dv(t))\\
	\breve{x}(t+1|t)	&= a\breve{x}(t)\\
						&= a\breve{x}(t|t-1) + aL(t)(c\xi(t)+dv(t))\\
	\xi(t+1) 			&= (a-aL(t)c)\xi(t) + bw(t) - aL(t)dv(t) 
\end{aligned}
\end{equation}
where $L(t)$ is the optimal Kalman filter gain(see, e. g., \cite{astrom:1970}): 
\begin{equation*} 
\begin{aligned}
L(t)					&= V_{\xi \xi}(t)c(c^2V_{\xi \xi}(t) + d^2V_{vv}(t) )^{-1} \\ 
V_{\xi \xi}(t+1) 	&=  (a-aL(t)c)^2 V_{\xi \xi}(t) \\
					&+
\begin{bmatrix}
	b & -aL(t)
\end{bmatrix}
\begin{bmatrix}
	V_{ww}(t) & V_{wv}(t)\\
	V_{vw}(t) & V_{vv}(t)	
\end{bmatrix}
\begin{bmatrix}
	b & -aL(t)
\end{bmatrix}^\intercal
\end{aligned}
\end{equation*}

We also know that $\gamma^{t-1}$ and $\xi(t)$ are uncorrelated according to Proposition \ref{lse}. This implies in turn that $y^{t-1}$ and $\xi(t)$ are uncorrelated.
Hence, the averaged estimation error of the decoder is equal to

$$
\frac{1}{T}\sum_{t=1}^T \E|x(t)-\hat{x}(t)|^2 = 
\frac{1}{T}\sum_{t=1}^T \left(\E|\breve{x}(t)-\hat{x}(t)|^2 + \E|\xi(t)|^2\right).
$$

Obviously, the decoder can't do much about the error covariance $\E|\xi(t)|^2$.
The decoder $\F$ minimizes the averaged estimation error above if and only if it minimizes the
averaged estimation error of $\breve{x}(t)$. 
Thus, we have transformed the output measurement problem to a state measurement 
problem at the encoder $\G$, where the measured state is the state $\breve{x}(t)$ 
of the linear time-varying dynamical system given by
\begin{equation*} 
\label{transs2}
	\begin{aligned}
	\breve{x}(t+1)		&= \breve{x}(t+1|t) + L(t+1)(c\xi(t+1)+dv(t+1))\\
						&= a\breve{x}(t) + L(t+1)(c\xi(t+1)+dv(t+1))\\
						&= a\breve{x}(t) + \beta(t) \omega(t)
	\end{aligned}
\end{equation*}
with $\omega(t) \sim \mathcal{N}(0,1)$ and
\begin{equation*} 
	\begin{aligned}
	\beta^2(t) &= L^2(t+1)\E\{(c\xi(t+1)+dv(t+1))^2\} \\
				&= L^2(t+1)  (c^2 V_{\xi \xi}(t+1)  + d^2 V_{vv}(t+1) )
	\end{aligned}
\end{equation*}
Inserting $b(t) = \beta(t)$ in Problem \ref{p1} and using
Theorem \ref{main} concludes the proof.

\end{document}